\documentclass[11pt]{article}

\usepackage{amsmath,latexsym,amssymb}

\usepackage{euscript}
\usepackage{eufrak}
\usepackage{amssymb}
\usepackage{mathrsfs}
\usepackage{multicol}

\newenvironment{proof}{\noindent\bf Proof. \rm}{\hfill $\mbox{\boldmath{$ \square$}}$}

\newenvironment{dem}{\noindent\bf Proof.\rm}{\hfill $\mbox{\boldmath{$ \square$}}$}

\newtheorem{coro}{\bf Corollary}
\newtheorem{theo}[coro]{\bf Theorem}
\newtheorem{prop}[coro]{\bf Proposition}
\newtheorem{defi}[coro]{\bf Definition}
\newtheorem{lem}[coro]{\bf Lemma}
\newtheorem{definition}[coro]{\bf Definition}

\def\S{\mathbf{S}}
\def\P{\mathcal{P}}
\def\F{\mathcal{F}}

\def\im{\!\rightarrow\!}

\title{ \large \textbf An algebraic study of the first order version of some implicational fragments of the three-valued \L ukasiewicz logic}

\author{Aldo Figallo-Orellano and Juan Sebasti\'an Slagter}
\date{}

\begin{document}
\maketitle

\begin{abstract}

MV-algebras are an algebraic semantics for \L ukasiewicz logic and MV-algebras generated by a finite chain are Heyting algebras where  the G\"odel implication can be written in terms of De Morgan and  Moisil's  modal operators. In our work, a fragment of trivalent \L ukasiewicz logic is studied. The propositional and first-order logic is presented. The maximal consistent theories are studied as  Monteiro's maximal deductive systems of the Lindenbaum-Tarski algebra, in both cases. Consequently, the adequacy theorem with respect to the suitable  algebraic structures is proven.
\end{abstract}

\section{Introduction and Preliminaries}

In 1923, David Hilbert proposed studying the implicative fragment of {\em classical propositional calculus}. This fragment is well-known as {\em positive  propositional calculus} and its study started in 1934 by D. Hilbert and P. Bernays. The following axiom schemas define this calculus

\begin{itemize}
  \item[(E1)] $\alpha \to (\beta \to \alpha)$,
  \item[(E2)] $(\alpha \to (\beta \to  \gamma))\to ((\alpha\to \beta)\to (\alpha \to \gamma))$.
\end{itemize}

and the inference rule {\em modus ponens} is

\begin{itemize}
  \item[(MP)] $\displaystyle \frac{\alpha, \alpha \to \beta}{\beta}$.
\end{itemize}

In 1950, L. Henkin introduced the {\em implicative models} as algebraic models of  the positive implicative calculus. Later, A. Monteiro  renamed them as {\em Hilbert algebras} and his Ph. D. student A. Diego (\cite{AD1}) made one of the most important contributions to these algebraic structures.

On the other hand, I. Thomas in \cite{IT} considered the $n$-valued positive implicative calculus, with signature $\{\to, 1\}$, as a calculus that  has a characteristic matrix $\langle A,\{1\}\rangle$ where $\{1\}$ is the set of designated elements and the algebra $A=(\mathbb{C}_n, \to, 1)$ is defined as follows

$$\mathbb{C}_n=\{0,\frac{1}{n}, \frac{2}{n}, \cdots, \frac{n-1}{n}, 1\}$$

and

$$x \to y= \begin{cases}1 & \text{if } x\leq y \\y & y<x \end{cases},$$

This author proved that for this calculus, we have to add the following axiom to the positive implicative calculus:

\begin{itemize}
\item[(E3)] $T_n(\alpha_0,\cdots ,\alpha_{n-1})=\beta_{n-2}\im(\beta_{n-3}\im(\cdots \im(\beta_0\im \alpha_0)\cdots))$, where
\item[] $\beta_{i}=(\alpha_i\im \alpha_{i+1})\im \alpha_0$ for all $i, 0\leq i\leq n-2$.
\end{itemize}

The algebraic counterpart of $n$-valued positive implicative calculus was studied by Luiz Monteiro in \cite{LM1} where the axiom (E3) is translated by the equation $T_n=1$ to Hilbert algebras. In particular, in the $n=3$ case, the variety is generated by an algebra that has this set $\mathbb{C}_3=\{0,\frac{1}{2}, 1\}$ as support and an implication $\to$ defined by the following table:

\begin{center}
\begin{tabular}{c | c  c  c}

$\to$ & $0$ & $\frac{1}{2}$ & $1$\\
\hline

$0$ & $1$ & $1$ &  $1$\\
$\frac{1}{2}$ & $0$ & $1$ &  $1$\\
$1$ & $0$ & $\frac{1}{2}$ & $1$
\end{tabular}\\
\vspace{0,3cm}
Table 1
\end{center}
 
It is clear that $3$-valued Hilbert algebras are Hilbert algebras that verify the following identity:

\begin{itemize}
\item[(IT3)] $((x\to y)\to z)\to (((z\to x)\to z)\to z) = 1$.
\end{itemize}

It is important to note that the implication defined in  table 1 characterizes the implicative fragment of $3$-valued G\"odel  logic.

On the other hand, infinite-valued \L ukasiewicz logic {\bf \L}, introduced for philosophical reasons by Jan \L ukasiewicz, is among the most important and widely studied of all non-classical logics. Later,  MV-algebras were introduced by C. Chang  in order to prove completeness with respect to the calculus {\bf \L}. These algebras are term equivalent to Wajsberg algebras. Besides, Komori introduced $CN$-algebra as algebraic models to {\bf \L} in terms of implication and negation. 

Recall that an algebra ${\cal A} = \langle A, \Rightarrow, \sim, 1\rangle$ is said to be a  Wajsberg algebra if it satisfies the following identities (see \cite{FRST,CDM}):

\begin{itemize}
  \item[(w1)] $1\Rightarrow x = x$, 
  \item[(w2)] $(x\Rightarrow y)\Rightarrow ((y\Rightarrow z)\Rightarrow (x\Rightarrow z)) =1$, 
  \item[(w3)] $ (x\Rightarrow y)\Rightarrow y = (y\Rightarrow x)\Rightarrow x$, 
  \item[(w4)] $(\sim y \Rightarrow \sim x)\Rightarrow (x\Rightarrow y) = 1$.
\end{itemize}

We can define  other operations in a Wajsberg algebra. Indeed, $1 = \sim 0$, $x \oplus y = \sim x\Rightarrow y  $,  $x\odot y = \sim(\sim x \oplus \sim y)$,  $x\vee y = \sim (\sim x \oplus y) \oplus y = (x\Rightarrow y)\Rightarrow y$, $x\wedge y = \sim (\sim x \vee \sim y) $, where $\wedge$ and $\vee$ are lattice-operations. If we consider the operations  $\oplus$ and $\odot$ as primitive operations, then we have $(A, \oplus, \odot, \sim, 0)$ is an MV-algebra in Chang's formulation. Conversely, any MV-algebra in Chang's formulation produces one Wajsberg algebra by the appropriate definitions of $\sim$ and $\Rightarrow$ (\cite{FRST}), where $x\Rightarrow y = \sim x \oplus y $. Besides, it is well-known that the category of MV-algebra is equivalent to the  category of $l$-groups with strong unit.

Let us remark that the variety of MV-algebras generated by an MV-chain of length $n< \omega$ is often denoted by  MV$_{n}$-algebras. This notion was introduced by Grigolia and can be axiomatized by adding two new axioms (w5) and (w6)  to the axioms of MV-algebras:

\begin{itemize}
  \item[(w5)] $x^{n-1} = x^n$,
  \item[(w6)] $n(x^j\oplus (\sim x \odot \sim x^{j-1}))=1$ for $1<j<n$ and $j$ does not divide $n$.
\end{itemize}

Besides, we denote by $C_{n}$ the MV$_{n}$-algebra whose universe is $\{0,\frac{1}{n}, \frac{2}{n}, \cdots, \frac{n-1}{n}, 1\} $ endowed with the operations $x\Rightarrow y:=min\{1,1-x+y\}$, $\sim x:=1-x$. Also, it is well known that if  $(A, \oplus ,\sim, 1)$ is an  MV$_{n}$-algebra then $\L_{n} (A)=\langle A, \wedge, \vee, \sim, \sigma_0, \ldots, \sigma_{n-1},0,1\rangle $ is an $n$-valued Lukasiewicz-Moisil algebra (see \cite{BFS}),  where the operators $\sigma_i:A\to A$ are lattice-homomorphisms,  for $1\leq i\leq n$, are defined in terms of the MV-operations and are called Moisil operators. On the other hand, it is well-known that the G\"odel implication can be written in terms of the Moisil and De Morgan's operations as follows
$$ x\Rightarrow y= x\vee \sim \sigma_{n-1} y \vee (\sigma_{n-1} y \wedge \sigma_{n-1} x \wedge \sim \sigma_n y) \vee \cdots  \vee (\sigma_1 y \wedge \sigma_1 x \wedge \sim \sigma_1 y)\vee ( \sigma_0 x \wedge  \sigma_0 y) (\ast).$$

This implication was discovered by Cignoli in his Ph. D. thesis, as we can see in the following papers \cite{RC0,RC,RC1}.
It is worth mentioning that Cignoli in \cite{RC1} studied the first-order $n$-valued  \L ukasiewicz logic. In that work, it was presented  the $n$-valued  \L ukasiewicz logic as extension of intuitionistic calculus, this fact was commented  in the abstract of \cite{RC}.  Indeed, Cignoli was based on the fact that MV$_{n}$-algebras can be defined in terms of symmetric Heyting algebras with Moisil operations adding a special set of operations (see \cite[Definition 2.1]{RC1}); that he called them  $n$-valued proper \L ukasiewicz algebras; that is to say, obviously, these  algebras are term equivalent to  MV$_{n}$-algebras.  The latter facts allowed him to present soundness and completeness Theorems for first-order  $n$-valued  \L ukasiewicz logic (\cite{RC1})  by means of the Rasiowa's technique for the standard models. Much more recently, Iorgulescu  studied the connection between  MV$_{n}$-algebras and $n$-valued Lukasiewicz-Moisil algebra in \cite{AI}.

On the other hand, in 1941, G. Moisil introduced $3$-valued \L ukasiewicz algebras (or $3$-valued \L ukasiewicz-Moisil algebras)  as algebraic models of  $3$-valued logic proposed by \L ukasiewicz (\cite{GrM}). It is well-known, and part of folklore, that the class of $3$-valued \L ukasiewicz algebras is term equivalent to the one of $3$-valued MV-algebras (see, for instance, \cite{BFS}). Recall that an algebra $(A, \wedge, \vee, \sim, \nabla, 0, 1)$  is a $3$-valued \L ukasiewicz algebras if the following conditions hold: (L0) $x\vee 1=1$, (L1) $x\wedge(x\vee y)=x$, (L2) $x\wedge(y\vee z)=(z\wedge x)\vee(y\wedge x)$, (L3) $\sim \sim x=x$, (L4) $\sim(x\wedge y)=\sim x\vee \sim y$, (L5) $\sim x \vee \nabla x=1$, (L6) $\sim x \wedge x=\sim x\wedge \nabla x$, and (L7) $\nabla(x\wedge y)=\nabla x\wedge \nabla y$.

 Besides, it is well-known  that the algebra  $(A, \wedge, \vee, \sim, 0, 1)$  is a De Morgan algebra if  (L0) to  (L4) hold (\cite[Definition 2.6]{BFS}). On the other hand,  the characteristic matrix of logic from trivalent \L ukasiewicz algebras has the operator  $\wedge$, $\vee$, $\sim$,   $\nabla$ (
possibility operator) and $\triangle$ (necessity operator) over the chain $\mathbb{C}_3=\{0,\frac{1}{2}, 1\}$, and they are defined by the next table:

\begin{center}
\begin{tabular}{c | c | c | c}

$x$ & $\sim x$ & $\nabla x$ & $\triangle x$\\
\hline

$0$ & $1$ & $0$ &  $0$\\
$\frac{1}{2}$ &  $\frac{1}{2}$ & $1$ & $0$\\
$1$ & $0$ & $1$ & $1$
\end{tabular}\\
\vspace{0,3cm}
Table 2
\end{center}

In addition, the implication  $\Rightarrow$ defined above can be obtained from the operator    $\wedge$, $\vee$, $\sim$,  $\nabla$ and $\triangle$ by the following formula:

$$x\Rightarrow y = \triangle \sim x \vee y \vee (\nabla \sim x \wedge \nabla y).$$

Besides, it is easy to check that  $\nabla x = ( x\to \triangle x) \to \triangle x$. From latter and the fact that the implication  can be written in terms of operations from $3$-valued \L ukasiewicz algebras,  the authors of \cite{FRS,FRS1} were motivated to study the interesting implicational fragments of the $3$-valued \L ukasiewicz logic. In general, for some technical aspects of  \L ukasiewicz-Moisil algebras,  the reader can consult \cite{BFS}.

The rest of the paper is organized as follows. In the next section, we introduce the class of modal $3$-valued Hilbert algebras with infimum, where the modal operator is the same considered by Moisil. Besides, we prove that variety of these algebras is semisimple and we determine the generating algebras. Later on, using our algebraic results, we present Hilbert Calculus has algebraic counterparts to these algebras introduced in this section.  In section 3, we introduce and study the class of modal $3$-valued Hilbert algebra with supremum and also, as application of our algebraic work, we present Hilbert calculus for the fragment with disjunction  soundness and completeness, in a strong version, with respect to this class of algebras. Finally, in Section 4,  we study the first order logic for the fragment with disjunction by means of an adaptation of the Rasiowa's technique (\cite{RA}) using our algebraic work for the propositional case.

For the sake of motivating these notes and roughly speaking, this work is developed using Henkin's notion of maximal consistent theory as Monteiro's maximal deductive system of Lindenbaum-Tarski algebra. Monteiro  named  it {\em Syst\`emes deductifs li\'es \`a ''$a$''} , where $a$ is an element of some given algebra such that the congruences are determined by deductive systems \cite[pag. 19]{AM}. We use  this notion, applying Monteiro's technique, in the Section \ref{weak} in order to prove this variety is semisimple and in the proof of the completeness. It is important to note that the relation of Henkin's maximal consistent theories and Monteiro's maximal deductive systems is only verified in some semisimple varities of algebras studied in Monteiro's school. In addition, for instance,  Nelson algebras, Heyting algebras, Hilbert algebras, residuated lattices, the implicational algebraic systems so-called {\em standard models} (considered in \cite{RA}) and others classes of algebras from non-semisimple varieties, this relation is not verified. This fact was one of our reasons for studying  the algebraic systems introduced in this note. In the Rasiowa's book, one can see the algebraic study of first-order of the logics of {\em standard models} and in order to present the algebraic models as models for these first-order logics. This work needs to prove the existence of the complete structures such as the Dedekind-Macneille completion for Boolean algebras or Heyting algebras in order to interpret the quantified forlumas. Using this method Cignoli needed to find the completation for $n$-valued \L ukasiewic algebras. In contrast, our technique simplifies the proof of the completeness theorem using the fact that the simple algebras are complete lattice. Moreover, we can apply the technique to the Cignoli's works, what is more,  it is possible to apply to several semisimple varieties of algebras studied in Monteiro's school. By the way, these observations will be part of the future works.


\section{\large Trivalent modal Hilbert algebras with infimum}

In this section, we  introduce trivalent modal Hilbert algebras with infimum, for short $iH^{\triangle}_3$-algebra. Using Monteiro's characterization of maximal congruences (see Definition \ref{Ligado}), we  prove that the variety of $iH^{\triangle}_3$-algebra is semisimple. Then, it will be presented a propositional calculi that has the class of  $iH_3$-algebra as algebraic counterparts.

 For the sake of brevity, in what follows, we only introduce those essential notions of {\em Hilbert algebras} that we need, thought not in full detail. Anyway, for more information about these algebras the reader can consult the bibliography. 
 
Now, recall that a Hilbert algebra is an algebra $(A,\to ,1)$ such that for all $x,y,z \in A$  verifies: 

(H1) $x \to (y \to x)=1$, 

(H2) $(x \to (y \to z)) \to ((x \to y) \to (x \to z))=1$, 

(H3) if $x \to y=1$, $y \to x=1$, then $x=y$. 

\

\noindent The following lemma is well-known

\begin{lem} \label{lema1cap1} Let $A$ be a Hilbert algebra. The following properties are satisfied for every $x,y,z\in A$:
\begin{enumerate}
  \item [\textup{(H4)}] If $x=1$ and $x\to y=1$, then $y=1$,
  \item [\textup{(H5)}] the  relation  $\leq$ defined by $x\leq y$ iff  $x\to y=1$, it is an order on  $A$ and  $1$ is the last element,
  \item [\textup{(H6)}] $x\to x=1$, \textup{(H7)} $x\leq y\to x$, \textup{(H8)} $x\to(y\to z)\leq(x\to y)\to(x\to z)$, \textup{(H9)} $x\to 1=1$,
  \item [\textup{(H10)}]  $x\leq y$ implies $z\to x\leq z\to y$, \textup{(H11)} $x\leq y\to z$ implies $y\leq x\to z$, 
  \item[\textup{(H12)}] $x\to((x\to y)\to y)=1$, \textup{(H13)} $1\to x=x$, 
  \item[\textup{(H14)}] $x\leq y$ implies $y\to z\leq x\to z$, \textup{(H15)} $x\to(y\to z)=y\to(x\to z)$,
  \item [\textup{(H16)}] $x\to(x\to y)=x\to y$, \textup{(H17)} $(x\to y)\to((y\to x)\to x)=(y\to x)\to((x\to y)\to y)$, 
  \item [\textup{(H18)}] $x\to(y\to z)=(x\to y)\to(x\to z)$, \textup{(H19)} $((x\to y)\to y)\to y=x\to y$.
\end{enumerate}
\end{lem}

The proof of last lemma can be found in \cite{FRST}. Now, recall that (\cite{FRS})

\begin{defi}\label{modal1} An algebra $(A,\to , \triangle, 1) $   is said to be a $3$-valued modal  Hilbert algebra if its reduct  $(A,\to,  1)$ is a $3$-valued Hilbert algebra  and $\triangle$ verifies the following identities: 

{\rm (M1)} $\triangle x \to x=1$,  

{\rm (M2)} $(( y \to \triangle y)\to (x \to \triangle\triangle x))\to \triangle (x\to y)= \triangle x \to \triangle\triangle y$, and 

{\rm (M3)} $(\triangle x\to \triangle y)\to \triangle x = \triangle x$. 

\noindent Besides, we define a new conective by $\nabla x = ( x\to \triangle x) \to \triangle x$.

\end{defi}

On the other hand, in \cite{AVF}, the authors introduced and studied the class of Hilbert algebras such that each pair of elements has infimum. Then,

\begin{defi} \label{def2.2}
An algebra $\langle A,\to,\wedge, 1\rangle $  is said to be an $iH_3$-algebra if the following conditions hold:
\begin{itemize}
  \item [\rm (1)] the reduct $\langle A,\to,1 \rangle$ is a Hilbert algebra such that the axiom (IT3)  is satisfied.
  
  \item [\rm (2)] the following identities hold: $(iH_1)$ $x\wedge(y\wedge z)=(x\wedge y)\wedge z$, $(iH_2)$ $x\wedge x=x$, $(iH_3)$ $x\wedge(x\to y)=x\wedge y$, and $(iH_4)$ $(x\to(y\wedge z))\to((x\to z)\wedge(x\to y))=1$.
 
 \end{itemize}
\end{defi}

Let us observe that all $iH_3$-algebra $A$ and every $x,y\in A$, we can define the supremum of $\{x,y\}$ in the following way:

$$x\vee y \overset{def}{=} ((x\to y)\to y) \wedge ((y\to x)\to x).$$

Indeed, let $a,b\in A$ and put $c=((a\to b)\to b)\wedge((b\to a)\to a)$. Since  $x\leq(x\to y)\to y$ and $x\leq(y\to x)\to x$ hold and  there exists the infimum  $((x\to y)\to y)\wedge((y\to x)\to x)$, then $c$ is upper bound of the set  $\{a,b\}$. Now, let us suppose  that $d$ is another upper bound of  $\{a,b\}$ such that  $c\nleq d$. Thus, there exists an irreducible deductive system  $P$ such that $c\in P$ and $d\notin P$ \cite[Corolario 1]{aD}. Besides, since  $a,b\leq d$ then  $a,b\notin P$. On the other hand, as  $A$ is a trivalent Hilbert algebra and according  to \cite[Th\'eor\`eme 4.1]{AM3}, we have $ a\to b\in P \,\, {\rm or} \,\, b\to a\in P.$
Now, if we suppose that $a\to b\in P$ and since $c\leq(b\to a)\to a$, then we can infer that  $a\in P$, which is a contradiction. If we consider the case $b\to a\in P$, we obtain again a contradiction. Thus, $c$ is the supremum of  $\{a,b\}$. Therefore,  all $iH_3$-algebra is a {\em relatively pseudocomplemented lattice} (see \cite{RA}), being as $x\wedge z\leq y$ iff $x\leq z\to y$. From the latter, we have  that each $iH_3$-algebra is a distributive lattice.

\begin{defi} \label{def2.2}
An algebra $(A,\to,\wedge, \triangle,1)$  is called a trivalent modal Hilbert algebra with infimum (for short, $iH^{\triangle}_3$-algebra) if the reduct $(A,\to,\wedge, 1)$ is an $iH_3$-algebra and the reduct $(A,\to, \triangle,1)$ is a $\triangle H_3$-algebra. 
\end{defi}

We note with $i\mathbb{H}^{\triangle}_3$ the variety  of  $iH^{\triangle}_3$-algebras.

\begin{lem}\label{lem2.1}
Let $A$ be a $iH^{\triangle}_3$-algebra. The following properties are satisfied for every $x,y,z\in A$:
\begin{itemize}
\item[1.] $x\leq y$ iff $x\to y=1$ iff $x\wedge y=x$, 2. $x\to (y \to z)= (x\wedge y)\to z$, 3. $x\to (x\wedge y)=x\to y$,
\item[4.] $(x\wedge y)\to(x\to y)=1$, 5. $(x\to y)\to( (z\wedge x)\to	 (z\wedge y))=1$, 6. $(x\wedge y)\to x=1$,
\item[7.] $(x\wedge y)\to y=1$, 8. $1\wedge x=x$, 9. $x\to(y\to(x\wedge y))=1$, 10. $\triangle 1=1$,

\item[11.] $\triangle (x\to y)\to (\triangle x \to \triangle y)=1$, 12. $\nabla (x\wedge y)= \nabla x \wedge \nabla y$, 
\item[13.] $\triangle (x\wedge y) = \triangle x \wedge \triangle y$, 14. $(\nabla x \to x) \wedge \nabla x = x$, 15. $x\to(x\wedge y)=x\to y$.
\end{itemize} 
\end{lem}
\begin{proof}
Taking into account  the fact  all trivalent Hilbert algebras are  distributive lattices and the Representation Theorem 2.6 from \cite{LM1}, we can prove  the condition (1.) to (9.). The rest of the proof follows the very definitions.
\end{proof}

\

\begin{defi}\label{mds}
For a given $iH^{\triangle}_3$-algebra  $A$ and  $D\subseteq A$. Then,  $D$ is said to be a deductive system if (D1) $1\in D$, and (D2) if $x, x\to y\in D$ imply $y\in D$. Besides, we say that $D$ is a modal if: (D3) $x\in D$ implies $\triangle x\in D$. 
\end{defi}
 
Given a $iH^{\triangle}_3$-algebra  $A$ and $\{H_i\}_{i\in I}$ a family of  modal deductive systems of $A$, then it is easy to see that $\bigcap\limits_{i\in I} H_i$ is a modal deductive system.  Thus, we can consider the notion of  modal deductive system generated by $H$, and we denote $[H)_m$, as an intersection of all modal deductive system $D$ such that $D\subseteq H$.  It is well-known that $[H)=\{x\in A: {\rm exist }\,\, h_1,\cdots, h_k\in H \,\, :\,\,h_1\to(h_2\to \cdots \to(h_k\to x)\cdots)=1\}$  where $k$ is a finite integer. Now, we will introduce the following notation:

\begin{itemize}
\item[] $(x_1, \ldots, x_{n-1};x_n)=\left\{ \begin{tabular}{ll}
								$x_n$  & if \, $n=1$ \\[3mm]
						$x_1 \im (x_2,\ldots,x_{n-1};x_n)$ & if \, $n>1$
					\end{tabular}\right.$.
\end{itemize}

Hence, we can write:
 
\begin{center}
$[H)=\{x\in A:$ there exist $h_1,\ldots,h_k \in D_1:$  $(h_1, \ldots, h_k;x)=1\}$.  
\end{center}

Then, we have the following result

\begin{prop} \label{propSDGM}
Let $A$ be a $iH^{\triangle}_3$-algebra-algebra, suppose that $H\subseteq A$ and $a\in A$. Then the following properties hold:
\begin{itemize}
  \item[\rm (i)] $[H)_m= \{x\in A: {\rm there \, exist }\,\, h_1,\cdots, h_k\in H \,\, :\,\, (\triangle h_1, \ldots, \triangle h_k;x)=1\}$,
  \item[\rm (ii)] $[a)_m=[\triangle  a)$, where  $[b)$ is the set $[\{b\})$,
  \item[\rm (iii)] $[H\cup \{a\} )_m = \{ x\in A: \triangle a\to x \in [H)_m\}$.
 
\end{itemize}
\end{prop}

\begin{proof}
It is a routine.
\end{proof}

Besides,  we denote by $D_m(A)$ the set of modal deductive systems of $iH^{\triangle}_3$-algebra $A$, and by $Con_{iH^{\triangle}_3} (A)$  the set of congruence relations of a given  $iH^{\triangle}_3$-algebra $A$.

\begin{lem} \label{DmdetCong}
For all $A \in i\mathbb{H}_3^{\triangle}$, we have that the poset $D_m(A)$ is lattice-isomorphic to  $Con_{iH^{\triangle}_3} (A)$. 
\end{lem}

\begin{proof}
It is well-known that the set of congruences of Hilbert algebra $A$ is lattice-isomorphic to the set of all deductive systems. This bijection is given for each deductive system $D$, we have the relation $R(D)=\{(x,y): x\to y, y\to x\in D\}$ is a congruence of $A$ such that the class of $1$ verifies $|1|_{R(D)}=D$. Besides, for each congruence $\theta$ of $A$ the class of $|1|_\theta$ is a deductive system and $R(|1|_\theta) =\theta$. From the latter and Lemma \ref{lem2.1} (10.) and (11.), we can infer that  every congruence $\theta$ for a given $A$ respect $\triangle$ and $|1|_\theta$ is a modal deductive system.
\end{proof}

\subsection{\large Weak deductive systems}\label{weak}

For each $iH^{\triangle}_3$-algebra $A$, we can define a new binary operation $\rightarrowtail$ named weak implication such that: $x\rightarrowtail y=\triangle x\to y$.

\begin{lem} Let $A\in i\mathbb{H}_3^{\triangle}$, for any $x,y,z\in A$ the following properties hold:
 
 \begin{itemize}
\item [\rm (wi1)] $1\rightarrowtail x=x$,
\item [\rm (wi2)] $x\rightarrowtail x=1$,
\item [\rm (wi3)] $x\rightarrowtail\triangle x=1$,
  
\item [\rm (wi4)] $x\rightarrowtail(y\rightarrowtail z)=(x\rightarrowtail y)\rightarrowtail(x\rightarrowtail z)$,
 \item [\rm (wi5)] $x\rightarrowtail(y\rightarrowtail x)=1$,
 \item [\rm (wi6)] $((x\rightarrowtail y)\rightarrowtail x)\rightarrowtail x=1$.
\end{itemize}
\end{lem}

Let  $A\in i\mathbb{H}_3^{\triangle}$ and suppose a subset $D\subseteq A$, we say that $D$ is a weak deductive system (w.d.s.) if $1\in D$, and $x,x\rightarrowtail y\in D$ imply $y\in D$. It is not hard to see that the set of modal deductive systems is equal to the set of  weak deductive systems. We denote by $\mathcal{D}_w(A)$  the set of weak deductive systems of a Hilbert algebra.

Now, for a given $iH^{\triangle}_3$-algebra $A$ and a (weak) deductive system $D$ of $A$ is said to be a maximal if for every (weak) deductive system $M$ such that $D\subseteq M$ implies $M=A$ or $M=D$. Besides, let us consider the set of all maximal w.d.s. $\mathcal{E}_w(A)$. A. Monteiro gave the following definition in order to characterize maximal deductive systems:

\begin{defi}(A. Monteiro)\label{Ligado}
Let $A$ be a  $iH^{\triangle}_3$-algebra, $D\in\mathcal{D}_w(A)$ and $p\in A$. We say that $D$ is a  weak deductive system tied  to $p$ if $p\notin D$ and for any $D'\in\mathcal{D}(A)$ such that $D\subsetneq D'$, then $p\in D'$.
\end{defi}

The importance for introducing the notion of weak deductive systems is to prove that every maximal  weak deductive system is a weak deductive system tied to some element of a given $iH^{\triangle}_3$-algebra, $A$. Conversely, and using  (wi6), we can prove every w.d.s is a maximal  weak deductive systems. Moreover, from (wi4), (wi5) and (wi1) and using A. Monteiro's techniques, we also can prove that $\{1\}=\underset{M\in\mathcal{E}_w(A)}{\bigcap}M$  and so, we have that the following lemma holds.

First, in what follows, it will be considered the quotient algebra $A/M$  defined  by $a\equiv_M b$ iff $a\to b, b\to a\in M$, see Lemma \ref{DmdetCong} and the canonical projection $q_{M}: A\to A/M$ defined by $q_{M}=|x|_M$ where $|x|_M$  denotes the equivalence class of $x$ generated by $M$.

\begin{lem}\label{SM1} Let $A$ be a $iH^{\triangle}_3$-algebra then map $\Phi:A\longrightarrow \underset{M\in\mathcal{E}_w(A)}{\prod}A/M$ such that $\Phi(x)(M)=q_{M}(x)$ is a one-to-one homomorphism; that is to say, the variety of $iH^{\triangle}_3$-algebras is semisimple.
\end{lem}
\begin{dem}
It is routine.
\end{dem}

\

The construction of the following homomorphism is fundamental to obtaining the generating algebras of the variety of  $iH^{\triangle}_3$-algebra. First, we denote by  $\mathbb{C}_{3}^{\to,\wedge}$ the $iH^{\triangle}_3$-algebra with support is a chain $\mathbb{C}_{3}$ ($0 <\frac{1}{2}<1$) and $\wedge$ is a lattice operation, and $\to$, $\triangle$ are defined in the Table 1 and 2 form Section 1. In addition, it is easy to see that  the algebra $\mathbb{C}_{3}^{\to,\wedge}$ has a unique subalgebra $\mathbb{C}_2^{\to,\wedge}=\langle \{0,1\},\to,\wedge,\triangle,1\rangle$.

\begin{theo}\label{TeoM0S} Let  $M$ be a non-trivial maximal modal deductive system of  $iH^{\triangle}_3$-algebra $A$. Let us consider the sets $M_{0}=\{x\in A: \nabla x \notin M\}$ and $M_{1/2}=\{x\in A: x\notin M, \nabla x\in M\}$, and the map $h:A\longrightarrow \mathbb{C}_{3}$ defined by\\
\centerline{$h(x)= \begin{cases}
0 & \mbox{if } x\in M_{0} \\
1/2 & \mbox{if } x\in M_{1/2} \\
1 & \mbox{if} x\in M.
\end{cases}$}
Then, $h$ is a $iH^{\triangle}_3$-homomorphism such that $h^{-1}(\{1\})=M$.
\end{theo}

\begin{dem} We shall prove only that $h(x\wedge y)= h(x)\wedge h(y)$, for the rest of the proof can be done in a similar manner. We will show:
\begin{itemize}
  \item [1)] if $x\in M_0$ and $y\in A$, then $x\wedge y\in M_0$,
  \item [2)] if $x\in M_{1/2}$ and $y\in A$, then $x\wedge y\in M_{1/2}$,
  \item [3)] if $x,y\in M$, then $x\wedge y\in M$,
  \item [] Indeed,
  \item [1)] Let $x\in M_0, y\in A$, then $x\notin M$, $\nabla x\notin M$. From Lemma \ref{lem2.1} (12.),  and the fact $x\notin M$, we have that $x\wedge y\notin M$. Since $(\nabla x\wedge\nabla y)\to\nabla x\in M$ but  $\nabla x\notin M$, $\nabla x\wedge \nabla y\notin M$, we infer that,  by Lemma \ref{lem2.1} (12.), $\nabla(x\wedge y)\notin M$ and therefore, $x\wedge y\in M_0$.
  \item [2)] Assume that $x\in M_{1/2}$ and $y\notin M_0$. Thus,  $x\notin M$, and $\nabla x, \nabla y\in M$ and then,  $x\wedge y\notin M$. From the latter and Lemma \ref{lem2.1} (9.), we can write  $\nabla x\to(\nabla y\to(\nabla x\wedge\nabla y))\in M$. Since $\nabla x, \nabla y\in M$ we have that $\nabla(x\wedge y)\in M$. So, $x\wedge y\in M_{1/2}$.
    \item [3)] It follows immediately from Lemma \ref{lem2.1} (9.), which completes the proof. 
\end{itemize}

\end{dem}

It is worth mentioning that the last theorem is an important tool for the algebraic study of the class of these algebras. Moreover, this also is an important tool to prove the completeness theorem for the associated logic. It is not possible to have this homomorphism  a general context such as Universal Algebra, we have to find it in order to show the generating algebras and to give the {\em canonical model}.  The defintion of this homomorphism is not the same for $3$-valued \L ukasiewiz algebras or for MV$_3$- algebras. By the way, according to Lemma \ref{SM1}  and Theorem \ref{TeoM0S}, and an adaptation of the first isomorphism theorem of Universal Algebra, we have proved the following theorem and corollary.

\begin{theo} The variety $i\mathbb{H}^{\triangle}_3$ is semisimple. Besides, the algebras   $\mathbb{C}_3^{\to,\wedge}=\langle \{0,\frac{1}{2},1\},\to,\wedge,\triangle,1\rangle$ and $\mathbb{C}_2^{\to,\wedge}=\langle \{0,1\},\to,\wedge,\triangle,1\rangle$ are the unique simple algebras.
\end{theo}


\subsection{ \large Hilbert calculus for  $iH^{\triangle}_3$-algebras}

In the sequel, we are going to exhibit a calculus $i\mathcal{H}_{\triangle}^{3}$. Now, consider the signature $\Sigma=\{\to, \wedge, \triangle \}$, and let $Var=\{p_1,p_2,\dots\}$ a numerable set of propositional variables. The propositional language generated by $\Sigma$ and $Var$ will be denoted by $\mathfrak{Fm}_i$. It is clear that  $\mathfrak{Fm}_i$ is the absolutely free algebra of formulas generated by $Var$.

\begin{definition} The calculus  $i\mathcal{H}_{\triangle}^{3}$ defined over the language $\mathfrak{Fm}_i$ is the Hilbert calculus obtained from the following axiom schemas and inference rules: 

\noindent {\bf{Axioms}}

\begin{itemize}
  \item [(A1)] $\alpha\to(\beta\to \alpha)$,
  \item [(A2)] $(\alpha\to(\beta\to \gamma))\to((\alpha\to \beta)\to(\alpha\to \gamma))$,
  \item [(A3)] $((\alpha\to \beta)\to \gamma)\to(((\gamma\to \alpha)\to \gamma)\to \gamma)$,
  \item [$(A_i 4)$] $(\alpha\wedge \beta)\to \beta$,  
  \item [$(A_i 5)$] $(\alpha\wedge\beta)\to\alpha$,
  \item [$(A_i 6)$] $\alpha\to(\beta\to(\alpha\wedge\beta))$,
  \item [$(A_i 7)$]  $\triangle \alpha\to \alpha$,
  \item [$(A_i 8)$] $\triangle(\triangle \alpha\to \beta)\to(\triangle \alpha\to\triangle \beta)$, 
  \item [$(A_i 9)$] $((\beta\to\triangle \beta)\to(\alpha\to\triangle(\alpha\to \beta)))\to\triangle(\alpha\to \beta)$, 
  \item [$(A_i 10)$] $((\triangle \alpha\to \beta)\to \gamma)\to((\triangle \alpha\to \gamma)\to \gamma)$.

\end{itemize}

\noindent Assume that $\nabla \alpha := (\alpha \to \triangle \alpha) \to \triangle \alpha$.

\noindent {\bf{Inference rules}}

(MP) $\dfrac{\alpha, \alpha\to\beta }{\beta}$ , (NEC) $\dfrac{\alpha}{\triangle\alpha}$, \, \, and $(R_\wedge)$ $\displaystyle\frac{\alpha\to\beta}{\alpha\to(\alpha\wedge\beta)}$.

\end{definition}

We are going to consider  the usual notion of derivation of a formula  $\alpha$ of $i\mathcal{H}_{\triangle}^{3}$, and we shall denote by $\vdash_i \alpha$. Now, let us consider the relation $\equiv_i$ defined by $\alpha \equiv_i \beta$ iff $\vdash_i \alpha \to \beta$ and $\vdash_i  \beta \to \alpha$. Then, we have the following technical result

\begin{lem}
The following properties and rules are verified in $i\mathcal{H}_{\triangle}^{3}$.

\begin{itemize}
 \item [$(M_i 1)$] $\vdash_i \alpha\to \alpha$, $(M_i 2)$ $\{\gamma\}\vdash_i {\alpha\to \gamma}$, $(M_i 3)$ $\{\alpha\to(\beta\to \gamma)\}\vdash_i {( \alpha \to \beta)\to( \alpha \to  \gamma )}$,
  \item [$(M_i 4)$] $\vdash_i ( \alpha \to(\beta\to  \gamma ))\to(\beta\to( \alpha \to  \gamma ))$, $(M_i 5)$ $\{ \alpha \to(\beta\to  \gamma )\}\vdash_i {\beta\to( \alpha \to  \gamma )}$, 
  \item [$(M_i 6)$] $\{ \alpha \to \beta\}\vdash_i {(\beta\to  \gamma )\to( \alpha \to  \gamma )}$, $(M_i 7)$ $\{ \alpha \to \beta\}\vdash_i {( \gamma \to  \alpha )\to( \gamma \to \beta)}$,
  \item [$(M_i 8)$] $\{ \alpha  \to \beta,\ \beta\to  \gamma \}\vdash_i  \alpha  \to  \gamma $, $(M_i 9)$ $\{ \alpha \equiv_i  \beta,  \theta \equiv_i   \eta \}\vdash_i {( \alpha \to  \theta )\equiv_i  (\beta\to  \eta )}$,
  \item [$(R6)_i$] $\{ \alpha \to \beta\}\vdash_i \{( \gamma \wedge  \alpha )\to( \gamma \wedge \beta)\}$, $(R10)_i$ $\{ \alpha \to \beta, \theta \to  \eta \}\vdash_i\{( \alpha \wedge  \theta )\to(\beta\wedge  \eta )\}$, 
  \item [$(M_i 10)$] $\{ \alpha \equiv_i  \beta,  \theta \equiv_i   \eta \}\vdash_i {( \alpha \wedge  \theta )\equiv_i (\beta\wedge  \eta )}$, $(M_i 11)$ $\vdash_i ( \alpha \to( \alpha \to \beta))\to( \alpha \to \beta)$,
  \item [$(M_i 12)$] $\vdash_i ((\alpha \to \beta)\to(\beta\to  \gamma ))\to( \alpha \to  \gamma )$, 
  \item [$(M_i 13)$] $((\alpha\to\beta)\to\beta)\to\beta\equiv_i\alpha\to\beta$, $(M_i 14)$ $\vdash_i (( \alpha \to \beta)\wedge \beta)\equiv_i  \beta$,
  \item [$(M_i 15)$] $\vdash_i ( \alpha \wedge( \alpha \to \beta))\equiv_i   \alpha \wedge \beta$, $(R11)_i$ $\{ \alpha \to \beta, \alpha \to  \gamma \}\vdash_i \{ \alpha \to(\beta\wedge  \gamma )\}$,
  \item [$(M_i 16)$] $\vdash_i ( \alpha \to(\beta\wedge  \gamma ))\to ( \alpha \to \beta)\wedge( \alpha \to  \gamma )$, $(M_i 17)$ $\vdash_i \triangle( \alpha \to  \beta )\to(\triangle  \alpha \to\triangle  \beta )$,
  \item [$(M_i 18)$] $\vdash_i \triangle  \alpha \equiv_i\triangle\triangle  \alpha $, $(M_i 19)$ $\vdash_i ((( \beta \to\triangle  \beta )\to( \alpha \to\triangle  \alpha ))\to\triangle( \alpha \to  \beta ))\to(\triangle  \alpha \to\triangle  \beta )$,
  \item [$(M_i 20)$] $\vdash_i ((( \beta \to\triangle  \beta )\to( \alpha \to\triangle\triangle  \alpha ))\to\triangle( \alpha \to  \beta ))\to(\triangle  \alpha \to\triangle\triangle  \beta )$, 
  \item [$(M_i 21)$] $\{ \alpha \to  \beta \}\vdash_i {\triangle  \alpha \to\triangle  \beta }$, $(M_i 22)$ $\{ \alpha \equiv_i   \beta \}\vdash_i {\triangle  \alpha \equiv_i \triangle  \beta }$, $(M_i 23)$ $\triangle(\alpha\to\triangle\alpha)\equiv_i\alpha\to\triangle\alpha$,
  \item [$(M_i 24)$] $\vdash_i (\triangle  \alpha \to\triangle\triangle  \beta )\to((( \beta \to\triangle  \beta )\to( \alpha \to\triangle\triangle  \alpha ))\to\triangle( \alpha \to  \beta ))$,
  \item [$(M_i 25)$] $\vdash_i ((( \beta \to\triangle  \beta )\to( \alpha \to\triangle\triangle  \alpha ))\to\triangle( \alpha \to  \beta )\equiv_i \triangle  \alpha \to\triangle\triangle  \beta $, 
  \item [$(M_i 26)$] $\vdash_i ((\beta\to\triangle\beta)\to\triangle\alpha)\to(((\alpha\to\triangle\alpha)\to\triangle\beta)\to\triangle(\alpha\to\beta))$, $(M_i 27)$ $\vdash_i \alpha\to\nabla\alpha$,
  \item [$(M_i 28)$] $\vdash_i\nabla\alpha\equiv_i(\alpha\to\triangle\alpha)\to\triangle\alpha$, $(M_i 29)$ $\vdash_i\nabla(\alpha\to\beta)\equiv_i\nabla\alpha\to\nabla\beta$, $(M_i 30)$ $\vdash_i\nabla\triangle\alpha\equiv_i\triangle\alpha$, 
  \item [$(M_i 31)$] $\nabla(\alpha\wedge\beta)\equiv_i\nabla\alpha\wedge\nabla\beta$.
\end{itemize}
\end{lem}
\begin{proof}
It is routine.
\end{proof}

The last result  gives the final ingredients to  obtain the soundness and completeness Theorem. Is is important to note we only need a few properties of this lemma but in order to prove these properties we need the rest of properties. 

\begin{lem}
The relation $\equiv_i $ is a congruence on $\mathfrak{Fm}_i$.
\end{lem}
\begin{dem} The relation $\equiv_i $ is reflexive,  symmetric and transitive follows immediately from $(M_i 1)$ and $(M_i 7)$. Let us suppose that $\alpha \equiv_i \beta$ and $\gamma \equiv_i \delta$ and taking into account $(M_i 6)$ and $(M_i 7)$, we can prove that  $\alpha\to \gamma \equiv_i \beta \to \delta$. Besides, suppose that $\alpha \equiv_i \beta$, the from $(M_i 17)$, (MP) and (NEC) we have that $\triangle\alpha \equiv_i \triangle\beta$.
\end{dem}

Since the $\equiv_i $ is a congruence, it allows to define the quotient algebra $\mathfrak{Fm}_i/\equiv_i$ that is so-called as Lindenbaum-Tarski algebra.

\begin{theo}\label{lindembaum2}  The algebra $\mathfrak{Fm}_i/\equiv_i$ is a $iH^{\triangle}_3$-algebra in which the operators are given as follows: $|\alpha|\to|\beta|=|\alpha\to \beta|$, $|\alpha|\wedge|\beta|=|\alpha\wedge \beta|$ and $|\beta \to \beta|=\{\alpha\in \mathfrak{Fm}_i:\vdash_i\alpha\}$ and where $|\delta|$ denotes the equivalence class of the formula $\delta$. Besides, $|\alpha|\leq|\beta|$ iff $\vdash_i\alpha\to \beta$.
\end{theo}
\begin{dem} It is easy to see that the relation ''$\leq$'' is an order relation on $\mathfrak{Fm}_i/\equiv_i$. Thus, it is clear that  $|\alpha|\leq|\beta\to \beta|$ for every $\alpha,\beta$ and so,  $|\beta\to \beta|$ is the last element of $\mathfrak{Fm}_i/\equiv_i$ that we denote with $1$ and also, it is not hard to see that $1=\{\alpha\in \mathfrak{Fm}_i:\vdash_i\alpha\}$.

On the other hand, taking into account  above remarks of $\leq$, the axioms (A1) and (A2), and  $(M_i 1)$, $(M_i 5)$, $(M_i 6)$ and $(M_i 7)$, we verify that conditions (H1), (H2) and (H3) are valid on  $\mathfrak{Fm}_i/\equiv_i$, see definition of Hilbert algebras of section 1. Also, form (A3) we can prove the condition (IT3) holds, see Definition \ref{def2.2}. Besides, taking $(A_i 4)$ $(A_i 5)$ $(A_i 7)$, we have the condition 2. of Definition \ref{def2.2} is verified.  Now, from $(A_i 7)$, $(A_i 17)$,  $(A_i 20)$, $(A_i 24)$ we can prove the axiom (M1) (M2) (M3) of Definition \ref{modal1} are verified for $\mathfrak{Fm}_i/\equiv_i$, which completes the proof.
\end{dem}

Remember that $\mathbb{C}_{3}^{\to,\wedge}$ is a $iH^{\triangle}_3$-algebra where the support is a chain $\mathbb{C}_{3}$ and $\wedge$ is a lattice operation, and $\to$, $\triangle$ are defined in the Table 1 and 2 from Section 1. In addition, for us a {\em logical matrix} for $i\mathcal{H}_{\triangle}^{3}$ is a pair $\langle \mathbb{C}_{3}^{\to,\wedge}, \{1\}\rangle$ where $\{1\}$ is the set of designated elements.

A function $v:\mathfrak{Fm}_i \to \mathbb{C}_{3}^{\to,\wedge}$ is a valuation for $i\triangle \mathbb{H}_3$ if it satisfies $v(\alpha\#\beta)=v(\alpha)\# v(\beta)$ with $\#\in \{\to, \wedge\}$, $v(\triangle \alpha)=\triangle v(\alpha)$. Besides, we say that $\alpha$ is {\em  valid semantically}  if $v(\alpha)=1$ for all valuation $v$  and, in this case, we denote $\vDash\alpha$.

\begin{theo}  {\rm (Soundness and Completeness Theorem)} 
Let $\alpha$ be a formula in  $\mathfrak{Fm}_i$. 

Then, $\vDash \alpha$ \ if and only if   \ $\vdash_i\alpha$.
\end{theo}
\begin{dem}
It is easy to see that every axiom in $i\mathcal{H}_{\triangle}^{3}$ is  valid semantically and the satisfaction  is preserved by the inference rules.
 
Conversely, let us suppose that $\alpha$ is  valid semantically; that is to say, for every valuation $v:\mathfrak{Fm}_i \to \mathbb{C}_{3}^{\to,\wedge}$ we have $v(v)=1$. 
On the other hand, according to Theorem \ref{lindembaum2}, we have $\mathfrak{Fm}_i/\equiv_i$ is a  $iH^{\triangle}_3$-algebra and now, consider the canonical projection $q: \mathfrak{Fm} \to \mathfrak{Fm}_i/\equiv_i$ defined by $q(\alpha)=|\alpha|$. From the latter and the hypothesis, $q(\alpha)=1$ which imply $\vdash_i\alpha$.
\end{dem}

\section{\large Trivalent modal Hilbert algebras with supremum}

In this section, we  introduce and study trivalent modal  Hilbert algebra with supremum, for short  $H^{\vee,\triangle}_3$-algebras. We are going to study the class of  $H^{\vee,\triangle}_3$-algebras in order to present a calculi sound and complete w.r.t.  class of  $H^{\vee,\triangle}_3$-algebras in propositional and first-order version. Our adequacy theorems are based on algebraic previous results and well-known results of Universal Algebra as the first isomorphism Theorem.  Now, consider the following

\begin{defi}\label{def5} An algebra  $\langle A,\to,\vee,\triangle,1\rangle$ is said to be a trivalent modal  Hilbert algebra with supremum (for short, $H^{\vee,\triangle}_3$-algebra) if the following properties hold:
\begin{itemize}
  \item [\rm (1)] the reduct $\langle A,\vee, 1\rangle$ is a join-semilattice with greatest element 1, and the conditions  {\rm (a)} $x\to(x\vee y)=1$ and {\rm (b)} $(x\to y)\to((x\vee y)\to y)=1$ hold. Besides, given $x,y\in A$ such that there exists the infimun of $\{x,y\}$, denote by $x\wedge y$, then $\triangle (x\wedge y)= \triangle x\wedge \triangle y$.
 \item [\rm (2)] The reduct  $\langle A,\to , \triangle, 1\rangle $ is a $\triangle H_3$-algebra.
\end{itemize}
\end{defi}

Next we are going to show some properties that  will be very useful for the rest of this section.

\begin{lem}\label{lat1} For a given $H^{\vee,\triangle}_3$-algebra $A$ and $a,b,c\in A$, then the following holds:
\begin{itemize}
  \item [($H^{\vee,\triangle}_3$1)] If $a\to b=1$, then $a\vee b=b$, ($H^{\vee,\triangle}_3$2) If $a\to c=1$ and $b\to c=1$, then $(a\vee b)\to c=1$,
  \item [($H^{\vee,\triangle}_3$3)] $a\to(a\vee b)=1$, ($H^{\vee,\triangle}_3$4) $(a\to c)\to((b\to c)\to ((a\vee b)\to c))=1$,
 
  \item [($H^{\vee,\triangle}_3$5)] $\triangle(a\vee b)=\triangle a\vee\triangle b$, ($H^{\vee,\triangle}_3$6) $\nabla(a\vee b)=\nabla a\vee\nabla b$.
\end{itemize}
\end{lem}
\begin{proof}
It is routine.
\end{proof}

\

For a given $H^{\vee,\triangle}_3$-algebra $A$, we are going to consider the notion of modal deductive system, see Definition \ref{mds}.

\begin{lem}\label{conS} Given a  $H^{\vee,\triangle}_3$-algebra $A$, there exists a lattice-isomorphism between the poset of congruences of $A$ and the poset of the modal deductive systems of $A$.
\end{lem}
\begin{proof}
The proof is similar to the Lemma \ref{DmdetCong}.
\end{proof}

\

The following lemma can be proved in a similar way that was made in the  Section \ref{weak} using the notion of weak deductive system and the notion of  deductive system tied to some element.

\begin{lem}\label{lem18} Let $A$ be a $H^{\vee,\triangle}_3$-algebra then map $\Phi:A\longrightarrow \underset{M\in\mathcal{E}_w(A)}{\prod}A/M$ such that $\Phi(x)(M)=q_{M}(x)$ is a homomorphism; that is to say, the variety of $H^{\vee,\triangle}_3$-algebras is semisimple.

\end{lem}

The construction of the following homomorphism is fundamental to obtaining the generating algebras of the variety of  $iH^{\triangle}_3$-algebra. Moreover, this homomorphism will play a central {\em role} in the adequacy theorems in a propositional  and first-order version of logic. First, we denote by  $\mathbb{C}_3^{\to,\vee}$ the $iH^{\triangle}_3$-algebra where the support is a chain $\mathbb{C}_{3}$ ($0 <\frac{1}{2}<1$) and $\vee$ is a lattice operation, and $\to$, $\triangle$ are defined in the Table 1 and 2 form Section 1. In addition, it is easy to see that  the algebra $\mathbb{C}_{3}^{\to,\wedge}$ has a unique subalgebra $\mathbb{C}_2^{\to,\wedge}=\langle \{0,1\},\to,\vee,\triangle,1\rangle$.

\begin{theo}\label{TeoM0S} Let  $M$ be a non-trivial maximal modal deductive system of  $H^{\vee,\triangle}_3$-algebra $A$. Let us consider the sets $M_{0}=\{x\in A: \nabla x \notin M\}$ and $M_{1/2}=\{x\in A: x\notin M, \nabla x\in M\}$, and the map $h:A\longrightarrow \mathbb{C}_{3}$ defined by\\
\centerline{$h(x)= \begin{cases}
0 & \mbox{if } x\in M_{0} \\
1/2 & \mbox{if } x\in M_{1/2} \\
1 & \mbox{if} x\in M.
\end{cases}$}
Then, $h$ is a homomorphism such that $h^{-1}(\{1\})=M$.
\end{theo}

\begin{dem} 
We shall prove only that $h(x\vee y)= h(x)\vee h(y)$, for the rest of the proof can be done in a similar manner.

\begin{itemize}
  \item [(1)] Let $x\in M$ and $y\in A$. Taking into account ($H^{\vee,\triangle}_3$3), we have that $x\to(x\vee y)=1$. Thus, from $D_1)$ and $D_2)$ then $x\vee y\in M$.
  \item [(3)] Let us consider $x,y\in M_0$ and suppose that $\nabla(x\vee y)\in M$, then by ($H^{\vee,\triangle}_3$6) we have that $\nabla x\vee \nabla y\in M$. Thus, according to $(H_3^{\vee,\triangle} 4)$ we infer that $(\nabla x\to\nabla x)\to((\nabla y\to\nabla x)\to((\nabla x\vee\nabla y)\to\nabla x))=1$. So,  from  $D_1)$,  $D_2)$ and (H6) we can obtain that $(\nabla y\to\nabla x)\to ((\nabla x\vee\nabla y)\to\nabla x)\in M$. Since $\nabla x\notin M$, we can infer that $\triangle\nabla y\to \nabla x\in M$ and so,  we have  $\nabla y\to\nabla x\in M$. Form the latter and  $D_2)$, we can write $(\nabla x\vee \nabla y)\to\nabla x\in M$. Therefore,  $\nabla x\in M$ which is impossible, then  $\nabla(x\vee y)\notin M$.

  \item [(4)] If $x\in M_0$ and $y\in M_{1/2}$,  since $\nabla y\to(\nabla x\vee\nabla y)=1$ and   $\nabla y\in M$ we can infer that  $\nabla x\vee\nabla y\in M$. Now, let us suppose that $x\vee y\in M$. From ($H^{\vee,\triangle}_3$4) we can write  $(x\to y)\to((y\to y)\to((x\vee y)\to y))=1$. Thus, $x\to y\in M$ and then, $y\in M$ which is a contradiction. Therefore, $x\vee y\in M_{1/2}$.
  \item [(5)] If $x\in M_{1/2}$ and  $y\in M_0$ we can prove that $x\vee y\in M_{1/2}$ in a similar way to (4).
  \item [(6)] Suppose that $x\in M_{1/2}$ and $y\in M_{1/2}$, then from  $(H_3^{\vee,\triangle} 6)$ we have that $\nabla (x\vee y)\in M$. On the other hand, let  us suppose  $x\vee y\in M$, thus by  $(H_3^{\vee,\triangle} 4)$ we infer that $(x\to x)\to((x\to y)\to((x\vee y)\to x))=1$. Hence, since  $x\to y\in M$ we can write $x\in M$ which is a contradiction. Therefore,  $x\vee y\in M_{1/2}$. 
\end{itemize}\end{dem}

According to Lemma \ref{lem18} and Theorem \ref{TeoM0S}, and well-known facts about universal algebra, we have proved the following theorem and corollary.


\begin{coro}\label{genalg} The variety of $H^{\vee,\triangle}_3$-algebras is semisimple. Besides, the algebras $\mathbb{C}_3^{\to,\vee}=\langle \{0,\frac{1}{2},1\},\to,\vee,\triangle,1\rangle$ and $\mathbb{C}_2^{\to,\vee}=\langle \{0,1\},\to,\vee,\triangle,1\rangle$ are the unique simple algebras.
\end{coro}

Let us notice that not every $H^{\vee,\triangle}_3$-algebra has infimum. To seeing that, it is enough to see some subalgebras of  $\mathbb{C}_3^{\to,\vee}\times \mathbb{C}_3^{\to,\vee}$ where $\times$ is the direct product.

\subsection{\large Propositional calculus for $H^{\vee,\triangle}_3$-algebras}

Let $\mathfrak{Fm}_s=\langle Fm, \vee,\to,\triangle\rangle$ be the absolutely free algebra over  $\Sigma=\{\to, \vee, \triangle \}$  generated by a set $Var=\{p_1,p_2,\cdots\}$ of numerable variables. Also, sometimes we say that $\mathfrak{Fm}_s$ is a language over $Var$ and $\Sigma$.  Consider now the following logic:

\begin{defi} We denote by  $\mathcal{H}_{\vee,\triangle}^{3}$ the Hilbert calculus determined by the following axioms and inference rules, where  $\alpha,\beta,\gamma,...\in Fm$:

{\bf{Axiom schemas}}
\begin{itemize}
  \item [(Ax1)] $\alpha\to(\beta\to\alpha)$,
  \item [(Ax2)] $(\alpha\to(\beta\to\gamma)\to((\alpha\to\beta)\to(\alpha\to\gamma))$,
  \item [(Ax3)] $((\alpha\to(\beta\to\gamma))\to(((\gamma\to\alpha)\to\gamma)\to\gamma)$,
  \item [(Ax4)] $\alpha\to(\alpha\vee\beta)$,
  \item [(Ax5)] $\beta\to(\alpha\vee\beta)$,
  \item [(Ax6)] $(\alpha\to\gamma)\to((\beta\to\gamma)\to((\alpha\vee\beta)\to\gamma))$,
  \item [(Ax7)] $\triangle\alpha\to\alpha$,
  \item [(Ax8)] $\triangle(\triangle\alpha\to\beta)\to(\triangle\alpha\to\triangle\beta)$,
  \item [(Ax9)] $((\beta\to\triangle\beta)\to(\alpha\to\triangle(\alpha\to\beta)))\to\triangle(\alpha\to\beta)$,
  \item [(Ax10)] $((\triangle\alpha\to\beta)\to\gamma)\to((\triangle\alpha\to\gamma)\to\gamma)$.

\end{itemize}

\noindent {\bf{Inference rules}}

(MP) $\dfrac{\alpha, \alpha\to\beta }{\beta}$, (NEC) $\dfrac{\alpha}{\triangle\alpha}$.

Assume that $\nabla \alpha := (\alpha \to \triangle \alpha) \to \triangle \alpha$.

\end{defi}

Let $\Gamma \cup \{\alpha \}$ be a set formulas of $\mathcal{H}_{\vee,\triangle}^{3}$, we  define the derivation of $\alpha$ from $\Gamma$  in usual way and denote by $\Gamma\vdash_\vee \alpha$. 

\begin{lem} The following rules are derivable in $\mathcal{H}_{\vee,\triangle}^3$:
\begin{itemize}
  \item [($P_s$1)] $\vdash_\vee\{(x\vee y)\to(y\vee x)\}$
  \item [($P_s$2)] $\{x\to y\}\vdash_\vee\{(x\vee z)\to(y\vee z)\}$
  \item [($P_s$3)] $\{x\to y,u\to v\}\vdash_\vee\{(x\vee u)\to(y\vee v)\}$
  \item [($R_\vee 3$)] $\dfrac{\alpha\to\beta}{(\alpha\vee\beta)\to\beta}$
\end{itemize}
\end{lem}
\begin{proof}
It is routine.
\end{proof}

Now, we denote by $\alpha\equiv_\vee\beta$ if conditions  $\vdash_\vee\alpha\to\beta$ and  $\vdash_\vee\beta\to\alpha$ hold. Then,

\begin{lem} $\equiv_\vee$ is a congruence on $\mathfrak{Fm}_s$.

\end{lem}
\begin{dem} We only have to prove that if  $\alpha \equiv_\vee \beta$ and $\gamma\equiv_\vee \delta$, then $\alpha\vee \gamma \equiv_\vee \beta \vee \delta$, which follows immediately from ($P_s$3).
\end{dem}

Since the $\equiv_\vee $ is a congruence, it allows to define the quotient algebra $\mathfrak{Fm}_s/\equiv_\vee$ that is so-called the Lindenbaum-Tarski algebra.

\begin{theo}\label{linP}  The Lindenbaum algebra $\mathfrak{Fm}_s/\equiv_\vee$ of $H^{\vee,\triangle}_3$ is a $H^{\vee,\triangle}_3$-algebra by defining: $|\alpha|\to|\beta|=|\alpha\to\beta|$, $|\alpha|\vee|\beta|=|\alpha\vee\beta|$ and $1=|\beta \to \beta|=\{\alpha\in\mathfrak{Fm}_s:\vdash_\vee\alpha\}$, where $|\delta|$ denotes the equivalence class of the formula  $\delta$.
\end{theo}

\begin{dem}
We only have to prove $\mathfrak{Fm}_s/\equiv_\vee$ is a join-semilattice and the axioms (a) and (b) from  Definition \ref{def5} (2). So, the first part follows from (Ax4), (Ax5) and (Ax6), and the second one follows from axioms (Ax4) and $(R_\vee3)$.
\end{dem}


\

On the other hand, let us remark that for the propositional calculus $i\mathcal{H}_{\triangle}^{3}$  is possible to define the conjunction connective by $\alpha\vee \beta := ((\alpha \to \beta)\to \beta) \wedge ((\beta\to \alpha )\to \alpha)$ (see \cite[pag. 170]{IT}). Thus, taking into account the Theorem \ref{lindembaum2}, it is easy to see that $\mathfrak{Fm}_i /\equiv_i$ is a Hilbert algebra with supremum where $|\alpha|\vee|\beta|=|\alpha\vee\beta|$ for every formula $\alpha$ and $\beta$. Therefore, the calculus $i\mathcal{H}_{\triangle}^{3}$ is a $\{\to,\wedge,\vee,\triangle\}$-fragment of a $3$-valued \L ukasiewicz logic where $\to$ is the $3$-valued G\"odel implication.

 Now, we are going to introduce some useful notions in order to prove a strong Completeness Theorem for $\mathcal{H}_{\vee,\triangle}^{3}$ w.r.t. the class of $H^{\vee,\triangle}_3$-algebras.

Recall that  a logic  defined over a  language ${\cal S}$ is a system $\mathcal{L}=\langle For, \vdash\rangle$ where $For$ is the set of formulas over ${\cal S}$ and the relation  $\vdash \subseteq {\cal P} (For) \times For$, ${\cal P}(A)$ is the set of all subsets of $A$. The logic $\mathcal{L}$ is said to be a tarskian if it satisfies the following properties, for every set $\Gamma\cup\Omega\cup\{\varphi,\beta\}$ of formulas:
\begin{itemize}	
  \item [\rm (1)] if $\alpha\in\Gamma$, then $\Gamma\vdash\alpha$,
  \item [\rm (2)] if $\Gamma\vdash\alpha$ and $\Gamma\subseteq\Omega$, then $\Omega\vdash\alpha$,
  \item [\rm (3)] if $\Omega\vdash\alpha$ and $\Gamma\vdash\beta$ for every $\beta\in\Omega$, then $\Gamma\vdash\alpha$.
\end{itemize}

\noindent A logic $\mathcal{L}$ is said to be finitary if it satisfies the following:

\begin{itemize}
  \item [\rm (4)] if $\Gamma\vdash\alpha$, then there exists a finite subset $\Gamma_0$ of $\Gamma$ such that $\Gamma_0\vdash\alpha$.
\end{itemize}

\begin{defi} \label{maxi} Let $\mathcal{L}$ be a tarskian logic and let $\Gamma\cup\{\varphi\}$ be a set of formulas, we say that $\Gamma$ is a theory. Besides,  $\Gamma$ is said to be a consistent theory if there is $\varphi$ such that $\Gamma\not\vdash_{\mathcal{L}}\varphi$. Besides, we say that $\Gamma$ is a maximal consistent theory if  $\Gamma,\psi\vdash_{\mathcal{L}}\varphi$ for any $\psi\notin\Gamma$ and in this case, we say $\Gamma$ non-trivial maximal respect to $\varphi$.
\end{defi}

A set of formulas $\Gamma$ is closed in $\mathcal{L}$ if the following property holds for every formula $\varphi$: $\Gamma\vdash_{\mathcal{L}}\varphi$ if and only if $\varphi\in\Gamma$. It is easy to see that any maximal consistent theory is a closed one.

\begin{lem}[Lindenbaum-\L os] \label{exmaxnotr} Let $\mathcal{L}$ be a tarskian and finitary logic. Let $\Gamma\cup\{\varphi\}$ be a set of formulas  such that $\Gamma\not\vdash\varphi$. Then, there exists a set of formulas $\Omega$ such that $\Gamma\subseteq\Omega$ with $\Omega$ maximal non-trivial with respect to $\varphi$ in $\mathcal{L}$.
\end{lem}
\begin{proof}
It can be found \cite[Theorem 2.22]{W}.
\end{proof}

\

It is worth mentioning that, by the very  definitions,   $\mathcal{H}_{\vee,\triangle}^{3}$ is a tarskian and finitary logic and then, we have the following

\begin{theo} \label{lig-maxim1} Let $\Gamma\cup\{\varphi\}\subseteq \mathfrak{Fm}_s$, with $\Gamma$ non-trivial maximal respect to $\varphi$ in $\mathcal{H}^3_{\vee,\triangle}$. Let $\Gamma/\equiv_\vee=\{\overline{\alpha}:\alpha\in\Gamma\}$ be a subset of the trivalent modal Hilbert algebra with supremum $\mathfrak{Fm}_s/\equiv_\vee$, then:
\begin{itemize}
  \item [1.] If $\alpha\in\Gamma$ and $\overline{\alpha}=\overline{\beta}$ then $\beta\in\Gamma$,
  \item [2.] $\Gamma/\equiv_\vee$ is a modal deductive system of $\mathfrak{Fm}/\equiv_\vee$. Also, if $\overline{\varphi}\notin\Gamma/\equiv_\vee$ and for any modal deductive system $\overline{D}$ which contains properly to $\Gamma/\equiv_\vee$, then $\overline{\varphi}\in\overline{D}$.
\end{itemize}
\end{theo}

\begin{dem} 
Taking into account  $\alpha\in\Gamma$ and $\alpha\equiv_\vee\beta$, we have that $\vdash\alpha\to\beta$ and $\vdash\beta\to\alpha$. Therefore, $\beta\in\Gamma$. Besides, it is not hard to see that $D_1)$, $D_2)$ and $D_3)$ are valid.

On the other hand, let $\overline{D}$ be mds that contains $\Gamma/\equiv_\vee$ and so, there is $\overline{\gamma}\in \overline{D}$ such that $\overline{\gamma}\notin\Gamma/\equiv_\vee$. Now, we have that $\gamma\notin\Gamma$ and therefore, $\Gamma\cup\{\gamma\}\vdash\varphi$. From the latter and taking into account $D=\{\alpha:\overline{\alpha}\in\overline{D}\}$ we can infer that $D\vdash\varphi$. Now, let us suppose that   $\alpha_1,...,\alpha_n$ is a derivation from $D$. We shall prove by induction over the length of the derivation that $\overline{\alpha_n}\in \overline{D}$:

If $n=1$ then $\alpha_1$ is an instance of an axiom or otherwise $\alpha_1\in D$. If  $\vdash\alpha_1$ is the case, then $\Gamma\vdash\alpha_1$ which is a contradiction. Then, we only have $\alpha_1\in D$ which implies  $\overline{\varphi}\in \overline{D}$.

Suppose that $\overline{\alpha_k}\in\overline{D}$ if $k$ is less than $n$. Then, we have the following cases: 

1. If $\varphi$ be the instance of an axiom, then $\Gamma\vdash\varphi$ which is a contradiction.

2. If $\varphi\in D$, then $\overline{\varphi}\in\overline{D}$.

3. If there exists $\{j,t_1,...,t_m\}\subseteq\{1,...,k-1\}$ such that $\alpha_{t_1},...,\alpha_{t_m}$  is a derivation of $\alpha_j\to\varphi$, then we have $\overline{\alpha_j\to\varphi}\in\overline{D}$ by induction hypothesis. So, $\overline{\alpha_j}\to\overline{\varphi}\in\overline{D}$. From the latter and  since $j<k$, we have $\overline{\alpha_j}\in\overline{D}$ and therefore, $\overline{\varphi}\in\overline{D}$.

4. If there exists $\{j,t_1,...,t_m\}\subseteq\{1,...,k-1\}$ such that $\alpha_{t_1},...,\alpha_{t_m}$ is a derivation of $\alpha_j$ and suppose that $\alpha_n$ is $\triangle\alpha_j$, then $\overline{\alpha_j}\in\overline{D}$. Now, since $\overline{D}$  is a mds, we have that $\triangle\overline{\alpha_j}\in\overline{D}$. Thus, $\overline{\varphi}\in\overline{D}$, which completes the proof.
\end{dem}

\

The notion of deductive systems considered in the last Theorem, part 2, was  named  {\em Syst\`emes deductifs li\'es \`a ''$a$''} by A. Monteiro, where $a$ is an element of some given algebra such that the congruences are determined by deductive systems \cite[pag. 19]{AM}.

Recall that for a given $iH^{\triangle}_3$-algebra $A$,  a {\em logical matrix} for $\mathcal{H}_{\vee,\triangle}^{3}$  is a pair $\langle A, \{1\}\rangle$ where $\{1\}$ is the set of designated elements. In addition, For a given $H^{\vee,\triangle}_3$-algebra $A$, we say that an homomorphism $v:\mathfrak{Fm}_s \to A$ is a valuation. Then, we say that $\varphi$ is a {\em semantical consequence} of $\Gamma$, and  we denote by $\Gamma\vDash_{\mathcal{H}^3_{\vee,\triangle}}\varphi$, if for every $H^{\vee,\triangle}_3$-algebra $A$ and every valuation $v$,  $v(\Gamma)=\{1\}$  then $v(\varphi)=1$. Beside, we say that $\alpha$ is valid in $A$ if $v(\alpha)=1$  for every valuation.

\begin{coro} \label{exvalua} Let $\Gamma\cup\{\varphi\}\subseteq\mathfrak{Fm}_s$, with $\Gamma$ non-trivial maximal respect to $\varphi$ in $\mathcal{H}^3_{\vee,\triangle}$. Then, there exists  a valuation  $v:\mathfrak{Fm}_s\to\mathbb{C}_3^{\to,\vee}$  such that $v(\varphi)=$  iff $\alpha\in \Gamma$.

\end{coro}

\begin{dem}
Taking into account Theorem  \ref{lig-maxim1}, we known that $\Gamma/\equiv_\vee$ is a maximal modal deductive system of $\mathfrak{Fm}_s/\equiv_\vee$. Then, by Theorem \ref{TeoM0S}, we have there is an homomorphism $h:\mathfrak{Fm}_s/\equiv_\vee \to \mathbb{C}_3^{\to,\vee}$ (see Corollary \ref{genalg}) such that $h^{-1}(\{1\})=\Gamma/\equiv_\vee$. Now, consider the canonical projection $\pi:\mathfrak{Fm}_s \to \mathfrak{Fm}_s/\equiv$ defined by $\pi(\alpha)=|\alpha|$, see Theorem \ref{linP}. Now, it is enough to take $v=h\circ \pi$.
\end{dem}

\begin{theo} \label{correcfue}(Soundness and completeness of $\mathcal{H}^3_{\vee,\triangle}$ w.r.t. $H^{\vee,\triangle}_3$-algebras) Let $\Gamma\cup\{\varphi\}\subseteq\mathfrak{Fm}_s$,  $\Gamma\vdash_\vee\varphi$ if and only if  $\Gamma\vDash_{\mathcal{H}^3_{\vee,\triangle}}\varphi$.
\end{theo}

\begin{dem}
Only if part (Soundness): It is not hard to see that every axiom is valid for every  $H^{\vee,\triangle}_3$-algebra $A$. In addition, satisfaction is preserved by the inference rules

If part (Completeness): Suppose  $\Gamma\vDash_{\mathcal{H}^3_{\vee,\triangle}}\varphi$ and $\Gamma\not\vdash_\vee\varphi$. Then, according to Lemma \ref{exmaxnotr}, there is maximal consistent theory $M$ such that $\Gamma\subseteq M$ and $M\not\vdash_\vee\varphi$. From the latter and Corollary \ref{exvalua}, there is a valuation $\mu: \mathfrak{Fm}_s\to \mathbb{C}_3^{\to,\vee} $ such that $\mu(\Delta)=\{1\}$ but $\mu(\varphi)\not=1$.
\end{dem}

\section{\large Model Theory and first order logics of  ${\cal H}^{\vee,\triangle}_3$ without identities}

In this section, we define the first order logic of $\mathcal{H}^{\vee,\triangle}_3 $.   Let $\Sigma=\{\to, \vee, \triangle \}$ be the propositional signature of $\mathcal{H}^{\vee,\triangle}_3$, the symbols $\forall$ (universal quantifier) and $\exists$ (existential quantifier), with the punctuation marks (commas and parentheses). Let $Var=\{v_1,v_2,...\}$ a numerable set of individual variables. A first order signature $\Theta$  is composed by the following elements:
\begin{itemize}
  \item a set $\mathcal{C}$ of individual constants, 
  \item for each $n\geq 1$, $\mathcal{F}$ a set of functions of arity $n$,
  \item for each $n\geq 1$, $\mathcal{P}$ a set of predicates of arity $n$.
\end{itemize}

The notions of bound and free variables inside a formula, closed terms, closed formulas (or sentences), and of term free for a variable in a formula are defined as usual. It will be denoted by $T_\Theta$ and $\mathfrak{Fm}_\Theta$  the sets of all terms and formulas, respectively. Given a formula $\varphi$, the formula obtained from $\varphi$
by substituting every free occurrence of a variable x by a term t will be denoted by $\varphi(x/t)$.

\begin{defi} Let $\Theta$ be a first order signature. The logic $\mathcal{QH}^{\vee,\triangle}_3$ over $\Theta$ is defined by Hilbert calculus obtained by extending $\mathcal{H}^{\vee,\triangle}_3$ expressed in the language $\mathfrak{Fm}_\Theta$ by adding the following:

{\bf Axioms Schemas}
\begin{itemize}
  \item [(Ax11)] $\varphi(x/t) \to\exists x\varphi$, if $t$ is a term free for $x$ in $\varphi$,
  \item [(Ax12)] $\forall x\varphi\to\varphi(x/t)$, if $t$ is a term free for $x$ in $\varphi$,

  \item [(Ax13)] $\triangle\exists x\varphi\leftrightarrow\exists x\triangle\varphi$,
  \item [(Ax14)] $\triangle\forall x\varphi\leftrightarrow\forall x\triangle\varphi$,
 
\end{itemize}

{\bf Inferences Rules}
\begin{itemize}
  \item [(R3)]$\dfrac{\varphi\to\psi}{\exists x\varphi\to\psi}$ where $x$ does not occur free in $\psi$,
  \item [(R4)]$\dfrac{\varphi\to\psi}{\varphi\to\forall x\psi}$ where $x$ does not occur free in $\varphi$.
\end{itemize}
\end{defi}

We denote by  $\vdash \alpha$ to a derivation of a formula $\alpha$ in $\mathcal{QH}_3^{\vee,\triangle}$ and  with $\Gamma\vdash\alpha$ to the derivation of $\alpha$ from a set of premises $\Gamma$. These notions are defined as the usual way. Besides, we denote $ \vdash\varphi\leftrightarrow \psi$ as an abbreviation of $ \vdash\varphi\rightarrow \psi$ and $ \vdash\varphi\rightarrow \psi$.

\begin{defi} \label{defim} Let $\Theta$ be a first-order signature. A $\Theta$-structure is a triple  $\mathfrak{S}=\langle A, S,\cdot^{\mathfrak{S}} \rangle$ such that $A$ is a complete $H^{\vee,\triangle}_3$-algebra, and $S$ is a non-empty set and $\cdot^\mathfrak{S}$ is an interpretation mapping defined on $\Theta$ as follows:
\begin{itemize}
  \item [1.] for each individual constant symbol $c$ of $\Theta$, $c^\mathfrak{S}$ of $S$,
  \item [2.] for each function symbol $f$ $n$-ary of $\Theta$, $f^\mathfrak{S}:S^n\to S$,
  \item [3.] for each predicate symbol $P$ $n$-ary of $\Theta$, $P^{\mathfrak{S}}:S^n\to A$.
\end{itemize}

\end{defi}

Given $\Theta$-structure $\mathfrak{S}=\langle A, S,\cdot^{\mathfrak{S}} \rangle$, a $\mathfrak{S}$-valuation is a function $v:Var \to S$. Given $a\in S$ and $\mathfrak{S}$-valuation $v$,  by $v[x\to a]$ we denote  the following $\mathfrak{S}$-valuation,  $v[x\to a](x)=a$ and $v[x\to a](y)=v(y)$ for any $y\in V$ such that $y\neq x$. 

 Let $\mathfrak{S}=\langle A, S,\cdot^{\mathfrak{S}} \rangle$ be a $\Theta$-structure and $v$ a $\mathfrak{S}$-valuation. A $\Theta$-structure $\mathfrak{S}=\langle A, S,\cdot^{\mathfrak{S}} \rangle$ and a $\mathfrak{S}$-valuation $v$ induce an interpretation map $||\cdot||^\mathfrak{S}_v$ for terms and formulas defined as follows

\begin{center}
$||x||^\mathfrak{S}_v=v(x)$,\\ [2mm]

$||c||^\mathfrak{S}_v=c^\mathfrak{A}$,\\ [2mm]

$||f(t_1,\cdots ,t_n)||^\mathfrak{S}_v=f_{\S}(||t_1||^\mathfrak{S}_v,\cdots ,||t_n||^\mathfrak{S}_v)$, for any $f\in\F$,\\ [2mm]

$||P(t_1,\cdots ,t_n)||^\mathfrak{S}_v=P_{\S}(||t_1||^\mathfrak{S}_v,\cdots ,||t_n||^\mathfrak{S}_v)$, for any $P\in\P$,\\ [2mm]
$||\alpha\to\beta||^\mathfrak{S}_v=||\alpha||^\mathfrak{S}_v\to ||\beta||^\mathfrak{S}_v$,\\ [2mm]
$||\alpha\vee\beta||^\mathfrak{S}_v=||\alpha||^\mathfrak{S}_v\vee ||\beta||^\mathfrak{S}_v$,\\ [2mm]
$||\triangle\alpha||^\mathfrak{S}_v=\triangle ||\alpha||^\mathfrak{S}_v$,\\ [2mm]
$||\forall x\alpha||^\mathfrak{S}_v=\underset{a\in S}{\bigwedge} ||\alpha||^\mathfrak{S}_{v[x\to a]}$,\\ [2mm]
$||\exists x\alpha||^\mathfrak{S}_v=\underset{a\in S}{\bigvee}||\alpha||^\mathfrak{S}_{v[x\to a]}$.

\end{center}

We say that $\mathfrak{S}$ and $v$ {\em satisfy} a formula $\varphi$, denoted by $\mathfrak{S}\vDash\varphi[v]$, if $||\varphi||^\mathfrak{S}_v=1$. Besides, we say that $\varphi$ is {\em true} $\mathfrak{S}$ if $||\varphi||^\mathfrak{S}_v=1$ for each  a $\mathfrak{S}$-valuation $v$ and we denote by $\mathfrak{S}\vDash\varphi$ . We say that $\varphi$ is a {\em semantical consequence} of $\Gamma$ in $\mathcal{QH}^{\vee,\triangle}_3$, if, for any structure $\mathfrak{S}$: if  $\mathfrak{S}\vDash\gamma$  for each $\gamma\in\Gamma$, then $\mathfrak{S}\vDash\varphi$. For a given set of formulas $\Gamma$, we say that the structure $\mathfrak{S}$  is a {\em model} of $\Gamma$ iff $\mathfrak{S}\vDash\gamma$  for each $\gamma\in\Gamma$. 

Now, it is worth mentioning the following property   $||\varphi(x/t)||^\mathfrak{A}_v= ||\varphi||^\mathfrak{A}_{v[x\to ||t||^\mathfrak{A}_v]}$ holds. Other important aspect of the definition of {\em semantical consequence} is that it is different to the propositional case because if one uses the definition of propositional case we are unable to prove an important rule as  $\alpha(x)\vDash \forall x \alpha(x)$. 

In addition, we need to exhibit some important property of complete $H^{\vee,\triangle}_3$-algebra.

\begin{lem}\label{LM}{\em \cite[Lemma 0.1.21]{LM}, see also \cite{PL}}
Let $A$ be a complete $H^{\vee,\triangle}_3$ and the set $\{a_i\}_{i\in I}$ of element of $A$ for a any non-empty set $I$. Then if there exist $\bigvee\limits_{i\in I} a_i$ ($\bigwedge\limits_{i\in I} a_i$) then there exist  $\bigvee\limits_{i\in I} \triangle a_i$    ($\bigwedge\limits_{i\in I} \triangle a_i$) and also, $\bigvee\limits_{i\in I} \triangle a_i = \triangle \bigvee\limits_{i\in I}  a_i $ and  $\bigwedge\limits_{i\in I} \triangle a_i = \triangle \bigwedge\limits_{i\in I}  a_i $.
\end{lem}

This property is useful to prove the following theorem.

\begin{theo} \label{correc} Let $\Gamma\cup\{\varphi\}\subseteq\mathfrak{Fm}_{\Theta}$, if $\Gamma\vdash_\vee\varphi$ then  $\Gamma\vDash\varphi$.

\end{theo}
\begin{dem} 
In what follows we will consider a fixed structure $\mathfrak{S}=\langle A, S,\cdot^{\mathfrak{S}} \rangle$. It is clear that the propositional axioms are true in $\mathfrak{S}$. Now, we have to prove the new axioms (Ax11) and (Ax12) are true in $\mathfrak{S}$, and the new inference rules (R3) and
(R4) preserve trueness in $\mathfrak{S}$.

(Ax11) Suppose that $\varphi$ is $\alpha(x/t)\to\exists x\alpha$. Then, $||\varphi||_v^\mathfrak{S}=||\alpha||_{v[x\to||t||_v^\mathfrak{S}]}^\mathfrak{S}\to ||\exists x\alpha||_v^\mathfrak{S}$. It is clear that $||\alpha||_{v[x\to||t||_v^\mathfrak{M}]}^\mathfrak{S}\leq\underset{a\in S}{\bigvee}||\alpha||_{v[x\to a]}^\mathfrak{S}$ and then, $||\alpha||_{v[x\to||t||_v^\mathfrak{S}]}^\mathfrak{S}\leq||\exists x\alpha||_v^\mathfrak{S}$. Therefore $||\alpha(x/t)\to\exists x\alpha||_v^\mathfrak{S}=1$ for every  $\mathfrak{S}$-valuation $v$. (Ax12) is analogous to (Ax11). Now, according to Lemma \ref{LM}, the axioms (Ax13) (Ax14) are true in $\mathfrak{S}$.

(R4)  Let $\alpha\to \beta$ such that $x$ is not free in $\alpha$, and let $\alpha\to\forall x\beta$. Let us suppose that $||\alpha \to \beta||_v^\mathfrak{S}=1$ for every  $\mathfrak{S}$-valuation $v$. Now, consider a fix valuation $v$ then   $||\alpha\to\forall x\beta||_v^\mathfrak{S}= ||\alpha||_v^\mathfrak{S}\to||\forall x\beta||_v^\mathfrak{S}=  ||\alpha||_v^\mathfrak{S}\to\underset{a\in S}{\bigwedge} ||\beta||^\mathfrak{S}_{v[x\to a]}$. On the other hand, By hypothesis we know that  $||\alpha||_u^\mathfrak{S} \leq ||\beta||_u^\mathfrak{S}$ for every $\mathfrak{S}$-valuation $u$. In particular,  $||\alpha||_v^\mathfrak{S}=||\alpha||_{v[x\to a]}^\mathfrak{S} \leq ||\beta||_{v[x\to a]}^\mathfrak{S}$ for every $\mathfrak{S}$-valuation $v$. Then, $||\alpha||_v^\mathfrak{S}\leq\underset{a\in S}{\bigwedge} ||\beta||^\mathfrak{S}_{v[x\to a]}$ and so, $||\alpha||_v^\mathfrak{S}\to\underset{a\in S}{\bigwedge} ||\beta||^\mathfrak{S}_{v[x\to a]}=1$ for every $\mathfrak{S}$-valuation $v$. The proof of preservation of trueness for (R3) is  analogous to (R4).
\end{dem}

\

In  what follows, we will prove a strong version of completeness Theorem for $\mathcal{QH}_{\vee,\triangle}^{3}$ using the Lindenbaum-Tarski algebra in a similar way the propositional case. Let us observe the algebra of formulas is an absolutely free algebra generated by the atomic formulas and its quantified formulas.

Now, let us consider the relation $\equiv$ defined by $\alpha\equiv \beta$ iff $\vdash \alpha \to \beta$ and $\vdash \alpha\to \beta$, then we have the algebra $\mathfrak{Fm}_{\Theta'}/\equiv$ is a $H^{\vee,\triangle}_3$-algebra and the proof is exactly the same as in the propositional case (see, for instance, \cite{Bell}). On the other hand, it is clear that $\mathcal{QH}_3^{\vee,\triangle}$ is a tarskian and finitary logic. So, we can consider the notion of (maximal) consistent and closed theories with respect to some formula  in the same way as the propositional case. Therefore, we have that Lindenbaum- \L os' Theorem holds for $\mathcal{QH}_3^{\vee,\triangle}$. Then, we have the following

\begin{theo} \label{lig-maxim} Let $\Gamma\cup\{\varphi\}\subseteq\mathfrak{Fm}_{\Theta}$, with $\Gamma$ non-trivial maximal respect to $\varphi$ in $\mathcal{QH}^{\vee,\triangle}_3$. Let $\Gamma/\equiv=\{\overline{\alpha}:\alpha\in\Gamma\}$ be a subset of  $\mathfrak{Fm}_\Theta/\equiv$, then:
\begin{itemize}
  \item [1.] If $\alpha\in\Gamma$ and $\overline{\alpha}=\overline{\beta}$, then $\beta\in\Gamma$.  Besides, it is verified that $\Gamma/\equiv\,=\,\{\overline{\alpha}: \Gamma \vdash\alpha \}$ in this case we say that it is closed. 
  
  \item [2.] $\Gamma/\equiv$ is a modal deductive system of $\mathfrak{Fm}_\Theta/\equiv$. Also, if $\overline{\varphi}\notin\Gamma/\equiv$ and  for any  modal deductive system $\overline{D}$ being  closed in the sense of 1 and  containing properly to $\Gamma/\equiv$, then $\overline{\varphi}\in\overline{D}$.
\end{itemize}
\end{theo}

\begin{dem}
According to the proof of Theorem \ref{lig-maxim1}, we only have to consider the rules $(R3)$ and $(R4)$. The fact that $\Gamma/\equiv$ is closed follows immediately.

In oder to complete the proof we have to consider two new cases. It is clear that  $\Gamma/\equiv$ a subset of $\overline{D}$. Now, let us consider $\overline{\phi}\in\overline{D}$ then $\overline{\phi}\notin\Gamma/\equiv$ and remember $D=\{\alpha:\overline{\alpha}\in\overline{D}\}$.
5. There exists $\{j,t_1,...,t_m\}\subseteq\{1,...,k-1\}$ such that $\alpha_{t_1},...,\alpha_{t_m}$ is a derivation of $\alpha_j=\theta\to\beta$. Let us suppose that $\alpha_n=\exists x\theta\to\beta$ is obtained by $\alpha_j$ applying $(R3)$. From induction hypothesis, we have that $\overline{\theta\to\beta}\in\overline{D}$.  From the latter, we obtain $\overline{\exists x\theta\to\beta}\in\overline{D}$.
6. There exists $\{j,t_1,...,t_m\}\subseteq\{1,...,k-1\}$ such that $\alpha_{t_1},...,\alpha_{t_m}$ is a derivation of $\alpha_j=\theta\to\beta$. Let us suppose that $\alpha_n=\theta\to\forall x\beta$ is obtained by $\alpha_j$ applying $(R4)$. From induction hypothesis, we have $\overline{\theta\to\beta}\in\overline{D}$ and then,  $\overline{\theta\to\forall x\beta}\in\overline{D}$.
\end{dem}

\

We note that for a given  maximal consistent theory $\Gamma$ of $\mathfrak{Fm}_{\Theta}$ we have  $\Gamma/\equiv$ is a maximal modal deductive system of $\mathfrak{Fm}_{\Theta}/\equiv$. If we denote  $A:=\mathfrak{Fm}_{\Theta}/\equiv$ and $\theta:= \Gamma/\equiv$ by well-known results of Universal algebras, we have the quotient algebra   $A/\theta$ is a simple algebra, see Corollary \ref{genalg}. From the latter and by adapting the first isomorphism theorem for Universal Algebras, we have that $A/\theta$ is an isomorphic to  $\mathfrak{Fm}_{\Theta}/\Gamma$ where it is defined by the congruence $\alpha \equiv_\Gamma \beta$ iff $\alpha\to \beta, \beta\to \alpha\in \Gamma$

\begin{theo} \label{correcfue1}  Let $\Gamma\cup\{\varphi\}$ be a set of sentences, then $\Gamma\vDash\varphi$ then $\Gamma\vdash\varphi$.
\end{theo}

\begin{dem} 
Let us suppose $\Gamma\vDash \varphi$ and $\Gamma\not\vdash\gamma$. Then, by Lindenbaum- \L os' Lemma, there exists $\Delta$ maximal consistent theory such that $\Gamma \subseteq \Delta$. Now, consider $\mathfrak{Fm}$ the algebra of closed formulas and  the algebra $\mathfrak{Fm}_{\Theta}/\Delta$  defined by the congruence $\alpha \equiv_\Delta \beta$ iff $\alpha\to \beta, \beta\to \alpha\in \Delta$. We know that $\mathfrak{Fm}_{\Theta}/\Delta$ is isomorphic to a subalgebra of $\mathbb{C}_3^{\to,\vee}$ and so, complete as lattice, in view of the above observations. Thus, taking the canonical projection $\pi_\Delta:\mathfrak{Fm} \to  \mathfrak{Fm}_{\Theta}/\Delta$.  

On the other hand, consider the structure $\mathfrak{M}=\langle \mathfrak{Fm}_{\Theta}/\Delta, T_{\Theta}, \cdot^{T_{\Theta}}\rangle $ where $T_{\Theta}$ is a set of terms. Then, it is clear that for every $t\in T_{\Theta}$ we have a constant $\hat{t}$ of $\Theta$. Now, we can consider a function $\mu:Var\to T_{\Theta}$ defined by $v(x)=x$. Besides, we have the interpretation $||\cdot||^\mathfrak{M}_\mu: \mathfrak{Fm} \to  \mathfrak{Fm}_{\Theta}/\Delta$ defined by if $\hat{t}$ is a constant then $||\hat{t}||^\mathfrak{M}_\mu:= t$, if  $f\in\F$ then $||f(t_1,\cdots ,t_n)||^\mathfrak{M}_\mu= f(t_1,\cdots ,t_n)$; if  $P\in\P$ then $||P(t_1,\cdots ,t_n)||^\mathfrak{M}_\mu= \pi_\Delta(P(t_1,\cdots ,t_n))$. Our interpretation is defined for atomic formulas; but it is easy to see that $||\alpha||^\mathfrak{M}_\mu=\pi_\Delta(\alpha)$ for every quantifier-free formula $\alpha$. Moreover, it is easy to see that for every formula $\phi(x)$ and every term $t$ we have $||\phi(x/\hat{t})||^\mathfrak{M}_\mu=||\phi(x/t)||^\mathfrak{M}_\mu$. Therefore, from the latter property and by (Ax12) and (R4), we have $||\forall x\alpha||^\mathfrak{A}_\mu=\underset{a\in T_{\Theta}}{\bigwedge} ||\alpha||^\mathfrak{A}_{\mu[x\to a]}$ and now using (Ax11) and (R3), we obtain $||\exists x\alpha||^\mathfrak{A}_\mu=\underset{a\in T_{\Theta}}{\bigvee}||\alpha||^\mathfrak{A}_{\mu[x\to a]}$. So, $||\cdot||^\mathfrak{M}_\mu$ is an interpretation map  such that $||\alpha||^\mathfrak{M}_\mu =1$ iff $\alpha \in \Delta$. On the other hand, it is not hard to see  for every   formula $\beta$, we have  $||\beta||^\mathfrak{M}_\mu=||\beta||^\mathfrak{M}_v$ for every $\mathfrak{M}$-valuation $v$. Therefore, $\mathfrak{M}\vDash\gamma$ for every $\gamma\in\Gamma$ but $\mathfrak{M}\not\vDash\varphi$.
\end{dem}

\

Given a formula $\varphi$ and suppose $\{x_1,\cdots,x_n\}$ is the set of variable of $\varphi$, the {\em universal closure} of $\varphi$ is defined by $\forall x_1\cdots \forall x_n\varphi$. Thus, it is clear that if $\varphi$ is a sentence then the universal closure of $\varphi$ is itself. Now, we are in condition to proving the following completeness theorem for formulas:
 
\begin{theo}  Let $\Gamma\cup\{\varphi\}$ be a set formulas, then $\Gamma\vDash\varphi$ then $\Gamma\vdash\varphi$.
\end{theo}
\begin{dem}
Let us suppose $\Gamma\vDash\varphi$ and consider the set $\forall\Gamma$ all universal closure of $\Gamma$. From the latter and definition of $\vDash$, we have $\forall\Gamma\vDash\forall x_1\cdots \forall x_n\varphi$. Then, according to Theorem \ref{correcfue1}, $\forall\Gamma\vdash\forall x_1\cdots \forall x_n\varphi$. Now, from latter and (Ax12) and (R4), we have $\Gamma\vdash\varphi$ as desired.
\end{dem}

\begin{theo}[Compactness Theorem] Let $\Omega$ be a subset of $\mathfrak{Fm}_\Theta$. $\Omega$ has a model if and only if any finite subtheory of $\Omega$ has a model.

\end{theo}



\

\





\begin{thebibliography}{99}
 
\bibitem{Bell} J. Bell and A. Slomson, {\em Models and Ultraproducts: An Introduction}, North Holland, Amsterdam, 1971
 
\bibitem{BFS} V. Boicescu and A. Filipoiu and G. Georgescu and S. Rudeanu, {\em \L ukasiewicz - Moisil Algebras}, Annals of Discrete Mathematics 49, North - Holland, 1991.

\bibitem{RC0} R. Cignoli, {\em Estudio algebraico de l\'ogicas polivalentes. Algebras de Moisil de orden $n$}, Ph. D. thesis, Universidad Nacional del Sur, Bahia Blanca, 1969. 


\bibitem{RC} R. Cignoli, {\em An algebraic approach to elementary theories based on $n$-valued \L ukasiewicz logics}. Z. Math. Logik Grundlag. Math. 30 (1984), no. 1, 87--96.

\bibitem{RC1} R. Cignoli, {\em Proper $n$-valued \L ukasiewicz algebras as $S$-algebras of \L ukasiewicz $n$-valued propositional calculi}. Studia Logica 41 (1982), no. 1, 3--16.


\bibitem{CDM} R. Cignoli, I.  D'Ottaviano and D. Mundici, {\em Algebraic foundations of many-valued reasoning}, Trends in Logic Studia Logica Library, 7. Kluwer Academic Publishers, Dordrecht, 2000. x+231 pp.

\bibitem{FRS} M. Canals Frau, A. V. Figallo and S. Saad, {\em Modal three valued Hilbert algebras}, Preprints Del ICB. Universidad Nacional  de San Juan, Argentina, p.1 - 21, 1990.

\bibitem{FRS1} M. Canals Frau and A. V. Figallo, {\em Modal 3-valued implicative semilattices}, Preprints Del Instituto de Ciencias B\'asicas. U. N. de San Juan, Argentina,  p.1 - 24, 1992. 




\bibitem{AD1} A. Diego, {\em Sur les alg\`ebres de Hilbert}, Coll\'ection de Logique Math\`ematique, ser. A, fasc. 21. Gouthier-Villars, Paris (1966)

\bibitem{AVF}  A. V. Figallo, G. Ram\'on and S. Saad, {\em A note on the distributive Hilbert algebras}, Proceedings of the Fifth ``Dr. Antonio A. R. Monteiro'' Congress on Mathematics, Bah\'{\i}a Blanca, (1999), 139--152. 

\bibitem{AVF1} A. V. Figallo, G. Ram\'on and S. Saad, {\em iH-Propositional calculus}, Bull. Sect. Logic Univ. L\'odz 35 (2006), no. 4, 157--162.

\bibitem{F1} A. Figallo Jr. and A. Ziliani, {\em Remarks on Hertz algebras and implicative semilattices}, Bull. Sect. Logic Univ. L\'odz, 34, 1 (2005), 37--42.

\bibitem{FRST} J. M. Font, A.J. Rodriguez and A. Torrens,{\em Wajsberg algebras}, Stochastica 8 (1984), Nro. 1, 5--31.


\bibitem{AI} A.  Iorgulescu, {\em Connections between MV$_n$-algebras and $n$-valued \L ukasiewicz--Moisil algebras Part I}, Discrete Math. 181, 155--177 (1998)
 
\bibitem{AM3} A. Monteiro, {\em Les alg\`ebras de Hilbert lin\`eaires, Unpublished papers I}, Notas de L\'ogica Matem\'atica, Univ. Nac. del Sur, Bah\'{\i}a Blanca, Vol. 40, (1996), 114--127.

\bibitem{AM} A. Monteiro, {\em Sur les alg\`ebres de Heyting simetriques}, Portugaliae  Math., 39, 1-4 (1980), 1--237.



\bibitem{AM2} A. Monteiro, {\em C\'alculo proposicional implicativo}, Imforme t\'ecnico N$^o$ 90. INMABB-Conicet, Universidad Nacional del Sur, 2005.
 

\bibitem{LM1} L. Monteiro, {\em Alg\`ebres de Hilbert $n-$valentes}, Portugaliae Math. 36(1977), 159--174.

\bibitem{LM} L. Monteiro, {\em Algebras de \L ukasiewicz trivalentes mon\'adicas}. Ph. D. thesis, Universidad Nacional del Sur, 1973.

\bibitem{GrM} Gr. C. Moisil, {\em Recherches sur logiques non-chrysippiennes}, Ann. Sc. de l'Universit\'e de Yassy, 27 (1941), 60-90.


\bibitem{IT} I. Thomas, {\em Finite limitations on Dummett's LC}, Notre Dame Journal of Formal Logic, 3 (1962), 170 -- 174.


\bibitem{PL} A. Petrovich and M. Lattanzi, {\em An alternative notion of quantifiers on three-valued Lukasiewicz algebras}, Mult.-Val. Logic Soft Comput.  28(4--5), 335--360 (2017)

\bibitem{RA} H. Rasiowa, {\em An algebraic approach to non-clasical logics}, Studies in logic and the foundations of mathematics, vol. 78. North-Holland Publishing Company, Amsterdam and London, and American Elsevier Publishing Company, Inc., New York, 1974.

\bibitem{W} R. W\'ojcicki, {\em Lectures on propositional calculi}, Ossolineum, Warsaw, 1984.

\end{thebibliography}
\end{document}